\documentclass[11pt,technote,a4paper,onecolumn,twoside]{IEEEtran} 

\usepackage{times}

\usepackage{moreverb}

\usepackage{amsmath, amssymb}

\usepackage{multicol}

\usepackage{epsfig}
\usepackage{cite}
\usepackage{url}

\usepackage{theorem}
\theorembodyfont{\normalfont}

\newtheorem{remark}{Remark}

\newtheorem{example}{Example}

\newtheorem{lemma}{Lemma}

\newtheorem{proposition}{Proposition}

\makeatletter 
\allowdisplaybreaks[4]
\def\QEDclosed{\mbox{\rule[0pt]{1.3ex}{1.3ex}}} 

\def\QED{\QEDclosed} 

\makeatother

\title{Effect of payload size on goodput
when message segmentations occur
for wireless networks:
Case of packet corruptions recovered by stop-and-wait protocol
}

\author{Takashi Ikegawa
\thanks{T.~Ikegawa is with Waseda Research Institute for Science and Engineering,
Waseda University and 
Graduate School of Mathematical Sciences, the University of Tokyo, Japan
(e-mail: ikegawa@aoni.waseda.jp or tikegawa@ms.u-tokyo.ac.jp).}
}

\begin{document}

\maketitle

\begin{abstract}

This paper investigates the effect of payload size on goodput 
for wireless networks
where packets created from a message through a segmentation function
are lost due to bit errors and they are recovered by a stop-and-wait protocol.
To achieve this,
we derive the exact analytical form of goodput
using the analytical form of a packet-size distribution,
given a message-size distribution and a payload size.
In previous work, the packet sizes are assumed to be constant, 
which are payload size plus header size,
although actual segmented packets are not constant in size.
Hence,
this constant packet-size assumption may be not justified for goodput analysis.
From numerical results,
we show that the constant packet-size assumption is not justified under low bit-error rates.
Furthermore,
we indicate that the curves of goodput are concave in payload size
under high bit-error rates.
In addition,
we show that the larger mean bit-error burst length yields less concave curves of goodput.

\end{abstract}


\section{Introduction}
\label{sec: Introduction}

In the past decades,
the evolution of wireless communication technologies,
such as widespread introduction of wireless local area networks (WLANs) and
emergence of cellular or Wi-Fi wide area networks,
changed our life style dramatically.
Wireless networks allow users to provide mobile communication service,
which is a more beneficial characteristic compared with wired networks.
However, the wireless networks 
have two main characteristics with negative impact on quality of service (QoS):
high bit-error rate and low link capacity.

For wireless networks,
which use open air as the transmission medium,
high bit-error rates ($10^{-2} - 10^{-6}$) are observed.
Hence,
wireless networks exhibit
unacceptable corruption probabilities of packets,
i.e., data units transferred over networks.

Wireless networks still provide lower link capacity,
compared to the fiber optical wired links.
The range of the link capacity for popular WLANs is $11 - 300$~Mbps.
Because many terminals share the same transmission medium,
the attainable capacity of each terminal becomes further smaller. 

To provide an error-free transmission service of messages,
i.e., data units generated by reliable applications such as transfer of Web pages and e-mail,
over the wireless networks,
solutions to overcome the high bit-error rate problem are required.
The straightforward solution is implementation of an error-recovery function for each terminal (or host),
which allows a sender to retransmit packets that are lost due to bit error and congestion.
For example,
IEEE 802.11 standard media access control protocol for WLANs \cite{CRO97} 
specifies a stop-and-wait protocol (SWP) 
to realize the error-recovery function 
in a simple manner.

Messages are frequently larger than the maximum permitted packet size, namely payload size.
To convey such messages over the network, 
a sender implements a message-segmentation function.
It enables the sender to divide a single message larger than the payload size into multiple packets.

The packet-size distribution has been significantly changed from the message-size distribution
through message segmentations.
For example,
papers \cite{FRA03, ALL00} show this tendency using actual traffic measurements.
On the other hand,
paper \cite{IKE12_PerEva} derives the analytical form of a packet-size distribution
when a message-size distribution and a payload size are given.
To achieve this,
in paper \cite{IKE12_PerEva} terms called {\it body} and {\it edge} packets are introduced.
The body packet is defined as a segmented packet appearing
between the head and penultimate packets in the original message.
The edge packet is the final segmented packet
if a message is segmented, or the message itself if it is not segmented. 
The sizes of body packets are equal to payload size plus header size,
whereas those of edge packets are variable,
not to exceed payload size plus header size.
From numerical results based on traffic measurement,
the findings include that the edge-packet occurrence probability is not negligible in some cases \cite{IKE12_PerEva}.

The packet size affects QoS measures such as mean response time and goodput
because link-level transmission delay and 
packet-corruption probability depends on the packet size.
For more details,
link-level transmission delay is simply given by the packet size divided by the link capacity \cite[pp.~64--65]{KUR16}
and the packet-corruption probability is approximately proportional to the packet size \cite[p.~132]{SCH87}.
Especially,
the packet size significantly affects these QoS measures
for high bit-error prone and/or low bandwidth links such wireless links.

The packet sizes are restricted to the payload size. 
Furthermore, the payload size is one of the controllable (manageable) parameters. 
If payload size is enough small,
the wireless networks are operating inefficiently
because an overhead such as header per packet transmission
is not negligible. 
On the other hand, 
if the packet size is larger, 
the packet is more likely to be corrupted,
necessitating more retransmissions and resulting in reduction of goodput. 
To solve this tradeoff issue,
many studies such as \cite{MOD99, VUR08, JEL08, LIN09, CIC11, DON14, KIM17}
proposed the adaptive payload sizing scheme.

In previous studies, however,
segmented-packet-size sequence behavior through the message-segmentation function,
has not been taken into consideration.
Thus,
the all packet sizes are assumed to be constant,
which are payload size plus header one,
although the edge-packet occurrence probability is assumed be negligible.

The purpose of this paper is to answer the following research questions:
\begin{description}

\item{Q1: }
What is approximation accuracy of
constant packet-size assumption?, and

\item{Q2:} What is relationship between payload 
size and goodput in the case of the burst bit-error occurrence
in addition to the independent one?

\end{description}
To achieve this,
we develop the exact analytical form of goodput
when segmented packets through the message-segmentation function
are lost due to bit error and they are recovered by an SWP.

The rest of the paper is organized as follows.
In the next section, we describe the communication network model
underlying our study.
In Section~\ref{sec: packet size sequence},
a model of a segmented packet size sequence is provided,
given the message-size distribution and the payload size.
Section~\ref{sec: goodput} derives the exact analytical form of goodput.
Section~\ref{sec: results} investigates the effect of payload size
on good for a simple scenario
where message sizes are constant.
Finally, Section~\ref{sec: conclusion} summarizes this paper and mentions future work. 


\section{Communication network model}
\label{sec: network model}

\begin{figure}
\centering
\epsfig{file=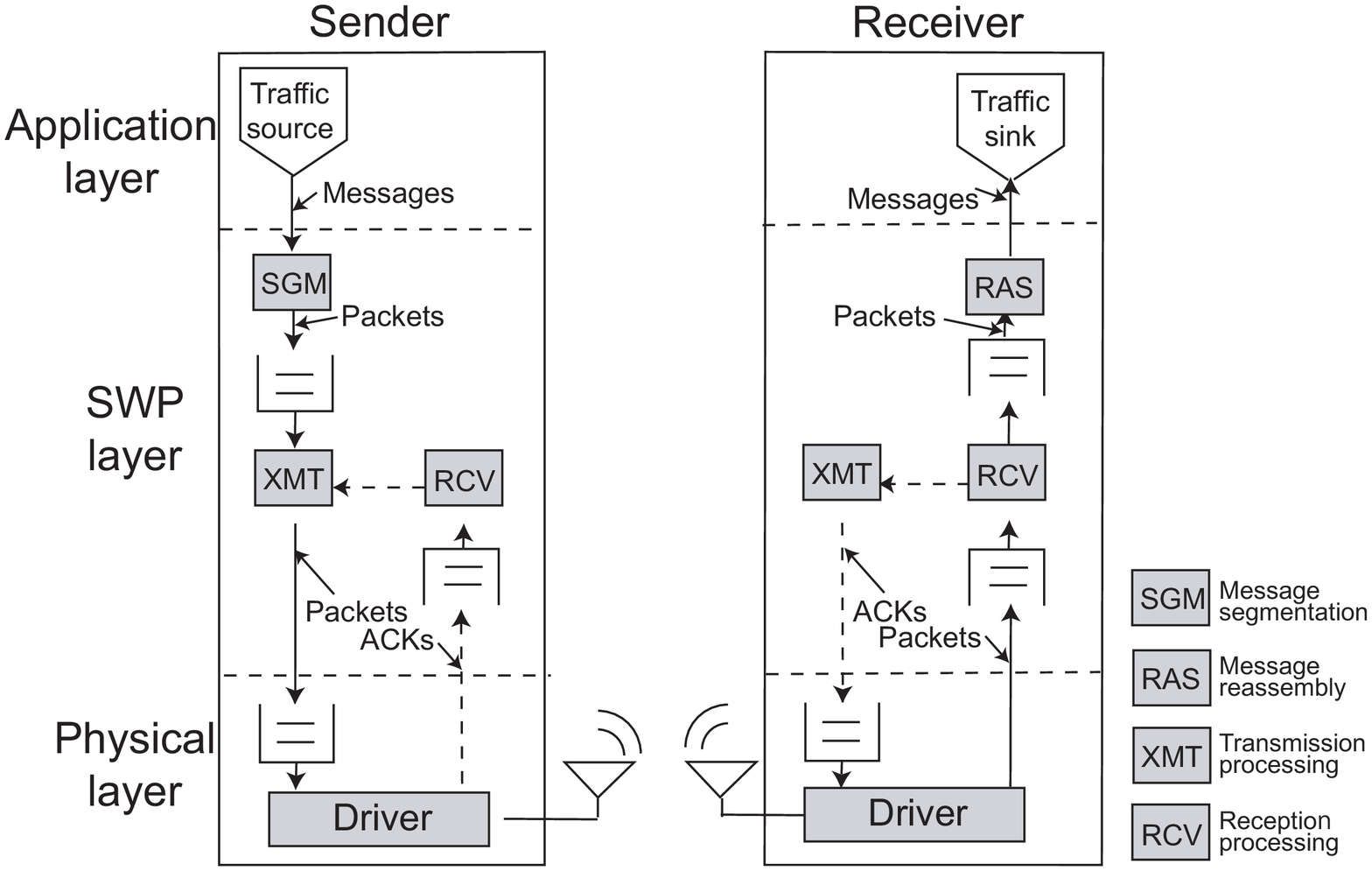, width=12cm} \\
\caption{Communication network model.}
\label{fig: network}
\end{figure}

In this section,
we first explain the two-layered communication network model
under consideration.
Next,
the model of data units introduced in this paper
at the respective layer is described.

\subsection{Layer model}

We consider a communication network
whose conceptual representation is shown in Fig.~\ref{fig: network}.
Each station (a sender and a receiver) has three layers.
The middle layer is referred to as an SWP layer.
It implements message segmentation-reassembly and error-recovery functions.
The error-recovery function is assumed to be implemented in a stop-and-wait scheme.
The layer above the SWP layer, namely the application layer, 
contains a traffic source and sink.
The traffic source generates the data units.
On the other hand,
the traffic sink terminates the corresponding data units.
The layer below the SWP layer,
namely the physical layer,
contains an entity that can transfer data units over wireless links at a sender.

\subsection{Data unit model}

We define data units exchanged between peer entities at the respective layer: messages and packets.

\begin{description}

\item{\bf Message:}
a data unit generated at a traffic source of the application layer
with a given size distribution.

\item{\bf Packet:}
a data unit created from a message through a segmentation-reassembly function,
and transferred over wireless links.
It consists of an information field and the appropriate protocol control information (PCI),
such as a header and/or trailer.
The information field contains a (divided) message.

\ \ The message-segmentation function implemented in the sender's segmentation-reassembly layer
enables a single message to be divided
into multiple packets
if the message size is larger than the
payload size, denote by $\ell^{(\rm{d})} (>0)$.
The receiver's segmentation-reassembly layer performs a message-reassembly function,
thus reassembling the segmented packets before delivering them
to the application layer.

\end{description}

\subsection{Burst bit error occurrence model}
\label{sec: burst bit error}

\begin{figure}
\centering{
\epsfig{file=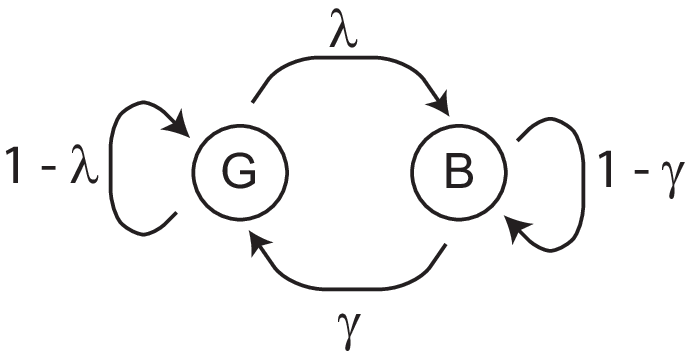, width=6cm} \\
\caption{Bit error model using two-state Markov chain.}
\label{fig: bit error model}
}
\end{figure}

Bit errors for WLANs
are reported to occur in bursts\cite{WIL02}.
%
%
One of the simple but capable of capturing the bit-error burstiness
is the Gilbert model \cite{GIL60}.
The process of the Gilbert model is a two-state Markov chain
with a transition diagram shown in Fig.~\ref{fig: bit error model}.
This Markov property enables us to obtain the form of goodput in a simple manner,
as will be mentioned in Section~\ref{sec: goodput}.

State $\rm{G}$ in this Markov chain represents the link in a ``good'' state,
indicating that it is operating at very low bit-error rate
denoted by $p^{(\rm{G})}$.
On the other hand,
state $\textrm{B}$ represents
the link is operating in a fading or shadowing condition 
at high bit-error rate, 
i.e., in a ``bad'' state,
denoted by $p^{(\rm{B})}$.
Within each state,
bit errors occur independently at their own rates.

Let $\upsilon(t) \in \mathcal{C}$
where $\mathcal{C} \stackrel{\triangle}{=} \{\textrm{G}, \textrm{B}\}$
be the link state at discrete time $t$ whose unit is bit.
We denote a transition rate from state $\textrm{G}$ to state $\textrm{B}$
by $\lambda$ where $0 \le \lambda \le 1$
and that from state $\textrm{B}$ to state $\textrm{G}$
by $\gamma$ where $0 \le \gamma \le 1$.
Then, the stochastic process $\{\upsilon(t); t =0, 1, 2, \cdots \}$ is represented as 
a discrete-time Markov chain with a transition matrix
\begin{align}
\mathbf{P}_\textrm{c} &\stackrel{\triangle}{=}
\left[\Pr\left(\upsilon(t+1)=\eta \, | \, \upsilon(t)=\xi\right), \eta \in \mathcal{C}, \xi \in \mathcal{C}\right] \notag \\
%
&= 
\begin{pmatrix}
1 - \lambda & \lambda \\
\gamma & 1 - \gamma
\end{pmatrix}.
\label{eq: matrix-P}
\end{align}

We let $\mathbf{\pi}_{\rm{c}}=(\pi^{(\rm{G})}, \pi^{(\rm{B})})$ be
the stationary probability vector of process $\{\upsilon(t)\}$,
i.e., the stationary probability vector of transition probability matrix $\mathbf{P}_{\rm{c}}$.
Then, its elements are 
given by
\begin{align}
\label{eq: pi}
\pi^{(\rm{G})} &= \cfrac{\gamma}{\lambda + \gamma} \quad\text{and}\quad 
\pi^{(\rm{B})} = \cfrac{\lambda}{\lambda + \gamma}.
\end{align}
Here, $\pi^{(\rm{G})}$ is the stationary state probability that the link is
in state $\textrm{G}$, and 
$\pi^{(\rm{B})}$ is that in state $\textrm{B}$
(usually $\pi^{(\rm{G})} > \pi^{(\rm{B})}$).

We define the \textit{mean} bit-error rate $p_{\rm{e}}$ as
\begin{align}
\label{eq: p_e}
p_{\rm{e}} &= 
\pi^{(\rm{G})} p^{(\rm{G})} + \pi^{(\rm{B})} p^{(\rm{B})},
\end{align}
assuming that $0\le p_{\rm{e}} <1$.

\bigskip

\begin{remark}{\it i.i.d. bit error model.}
\label{remark: i.i.d. bit error model}
Assume that bit errors occur independently
with a fixed bit-error rate $p_{\rm{e}}$.
This \textit{i.i.d.}~bit error model
is identical to 
the burst bit-error model
in the following cases:
$\lambda = 0$, $\gamma = 0$, or $\lambda + \gamma = 1$.
Thus, the burst bit-error model introduced above
accommodates the i.i.d.~model.
Actually,
the results from the correlated bit-error model
in these cases
agree with those from the i.i.d.~model
(see Proposition~\ref{pro: pr N L i i d}).
\hspace*{\fill}~\QED
\end{remark}




\subsection{Assumptions}

For analytical tractability,
we make the following assumptions.

\begin{description}

\item[\texttt{A1}:]
Message sizes are mutually independent and 
identically distributed 
according to a common message-size distribution function $F^{(m)}(\cdot)$.
The distribution $F^{(m)}(\cdot)$ has a finite mean value $\ell^{(m)}$,
which is referred to as the mean message size.

\item[\texttt{A2}:]
When a sender cannot receive a reply from a receiver, i.e., acknowledgement,
within in a specific {\it constant} interval timeout
after transmission,
it retransmits the lost packet.

\ We use notation $t_{\rm{out}}$ to express the value of the interval timeout in sec,
whereas $\hat{t}_{\rm{out}}$ which is expressed in bits, instead of $t_{\rm{out}}$,
will be used for goodput analysis in Section~\ref{sec: goodput}.
In this paper, the value of $t_{\rm{out}}$ is given by $t_{\rm{out}} \, \mu_{\rm{c}}$
where $\mu_{\rm{c}}$ is the capacity of the wireless link.

\item[\texttt{A3}:] 
We denote size of the acknowledgement packet, that is ACK, by $\ell^{(\rm{ACK})}$.

\item[\texttt{A4}:] 
The maximum number of retransmission attempts of the same packet is infinite.

\item[\texttt{A5}:]
The sender operates under a heavy traffic assumption,
which implies that the sender's SWP layer always has a packet
available to be sent.

\item[\texttt{A6}:]
the size of PCI is constant
and equal to $\ell^{({\rm h})}$.

\item[\texttt{A7}:]
the packet of the first transmission arrives at the link in state $\textrm{G}$ 
with probability $\pi_\textrm{G}$ and in state $\textrm{B}$ with probability $\pi^{(\rm{B})}$.

\item[\texttt{A8}:]
the sum of processing delay and propagation delay, which is independent of a packet size,
is constant and is denoted by $t_{\rm{pro}}$.

\end{description}

\medskip

\section{Packet size sequence model}
\label{sec: packet size sequence}

\begin{figure}
\centering{
\epsfig{file=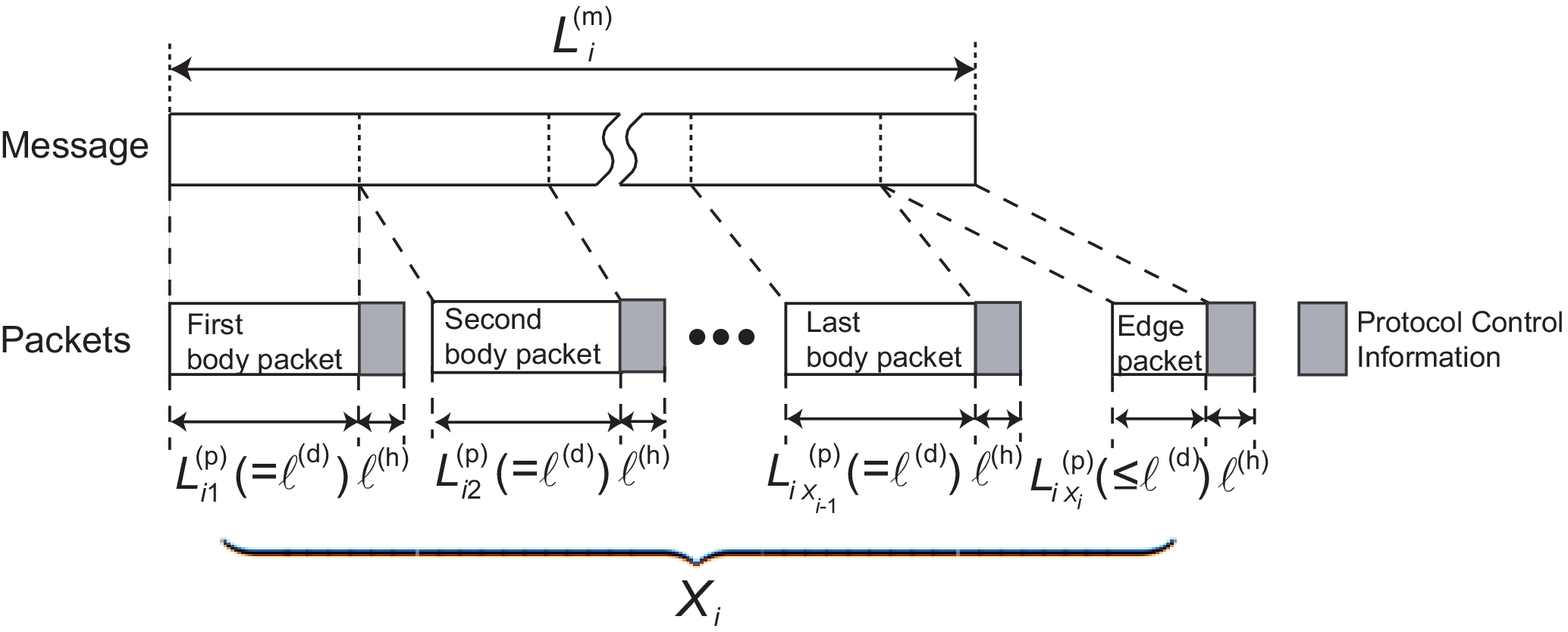, width=10cm} \\
\caption{Packets created from message whose size is $L_i^{({\rm m})}$.}
\label{fig: message packet}
}
\end{figure}

The creation of packets from a message through message segmentation is shown in Fig.~\ref{fig: message packet}.
As shown in Fig.~\ref{fig: message packet},
if $L_i^{({\rm m})} > \ell^{(\rm{d})}$,
then the $i$~th message is divided into multiple packets.
Denoting the size of an information field of the $j$~th packet
for the $i$th message by $L^{({\rm p})}_{i, j}$,
where $i=1, 2, \cdots,$ and $j = 1, 2, \cdots, X_i$,
we have
\begin{align}
\label{eq: L_p_i}
L^{({\rm p})}_{i, j} &= 
\begin{cases}
\ell^{(\rm{d})}, & \text{$j=1, 2, \ldots, X_i-1$,} \\
L^{({\rm m})}_i - (X_i-1) \, \ell^{(\rm{d})}, & \text{$j = X_i$},
\end{cases} 
\end{align}
with $X_i = \lceil L_i^{({\rm m})}/{\ell^{(\rm{d})}} \rceil$.
Here, the operator $\lceil a \rceil$ represents the smallest integer that is greater than or equal to $a$. 

We categorize a set of packets for the $i$th message into two kinds: body and edge packets.
we refer to a segmented packet appearing between the head and the penultimate packets 
i.e., packet from first to $(X_i - 1)$~st as a {\it body} packet, and
the final packet, i.e., $X_i$~th packet as an {\it edge} packet.

If $L_i^{({\rm m})} \le \ell^{(\rm{d})}$,
then the $i$th message is not segmented,
and a single packet, whose information field is identical to the original message,
is generated.
We also refer to this as edge packet
because it satisfies \eqref{eq: L_p_i}.
We note that the size of body packets are always equal to $\ell^{(\rm{d})} + \ell^{(\rm{h})}$, 
whereas that of edge packets is variable but does not exceed $\ell^{(\rm{d})} + \ell^{(\rm{h})}$.

\begin{remark}
The edge-packet size is less than {\it or equal to} $\ell^{(\rm{d})} + \ell^{(\rm{h})}$.
For example,
edge-packet sizes equal $\ell^{(\rm{d})} + \ell^{(\rm{h})}$
when message sizes are a multiple of $\ell^{(\rm{d})}$
although message segmentations happen
(see Section~\ref{sec: validation}).
\hspace*{\fill}~\QED
\end{remark}

We construct a stochastic process $\{L^{(p)}_\kappa: \kappa = 1, 2, \cdots\}$,
replacing the pair of epoch labels $i j$ with an in-sequence number $\kappa = 1, 2, \cdots$
for $\{L^{(p)}_{i \, j}\}$ in \eqref{eq: L_p_i}.

\begin{figure}
\centering{
\epsfig{file=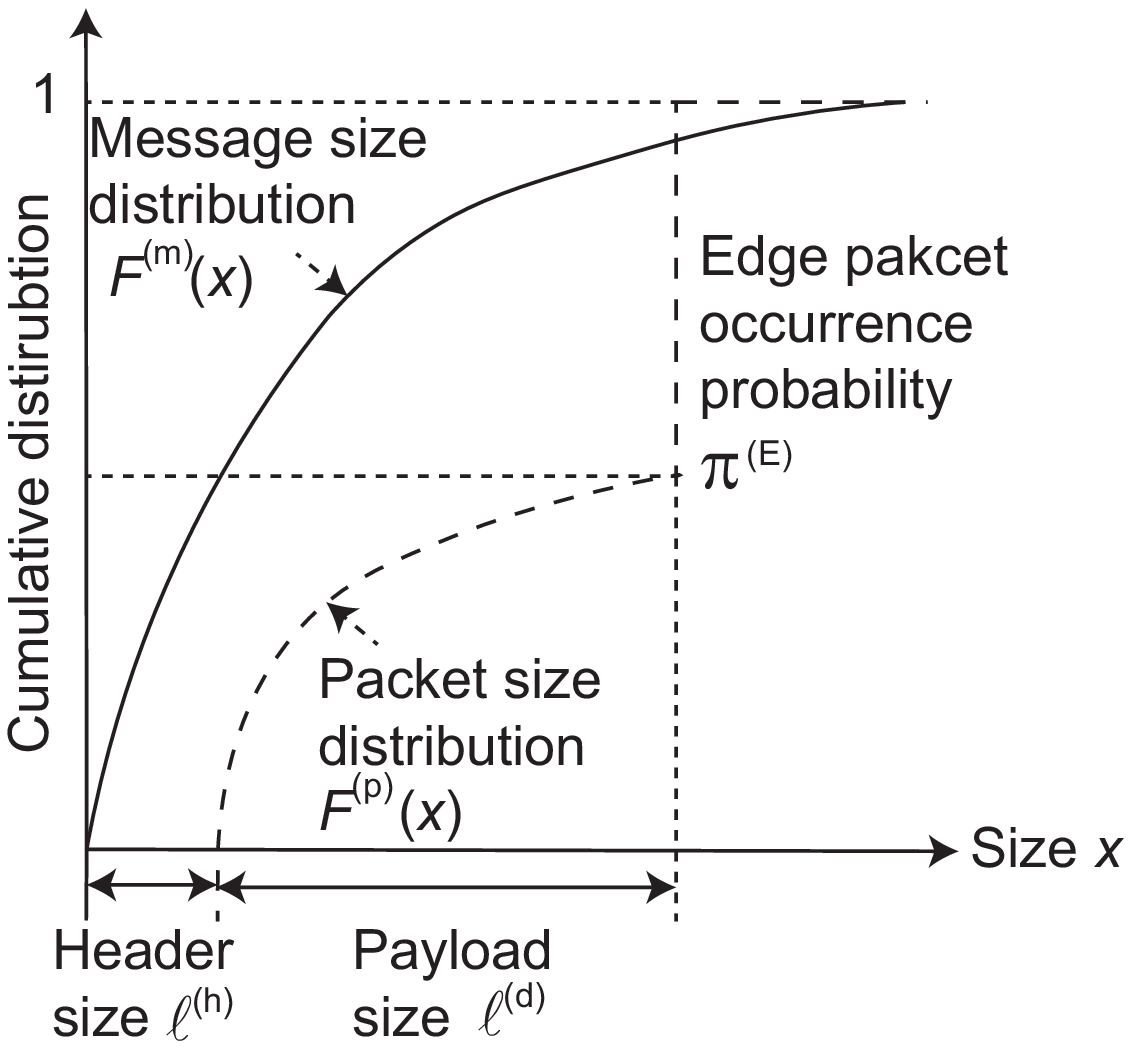, width=6cm} \\
\caption{Sketch of packet size distribution.}
\label{fig: packet size}
}
\end{figure}

Denoting the stationary packet size distribution by $F^{(\rm{p})}(\cdot)$,
the assumptions of \texttt{A1} and \texttt{A2}, using the argument \cite{IKE12_PerEva},
we have
\begin{align}
\label{eq: F^p}
F^{(\rm{p})}(x)
&\stackrel{\triangle}{=} \Pr\left(L^{(\rm{p})}_\kappa \le x\right)
= \left(1 - \pi^{(\rm{E})}\right) \textbf{1}\left(x - (\ell^{(\rm{d})} + \ell^{(\rm{h}))}\right)
+ \pi^{(\textrm{E})} F^{(\rm{E})}(x).
\end{align}
where $\pi^{(\rm{E})}$ is an edge-packet occurrence probability,
and $F^{(\rm{E})}(\cdot)$ is a stationary edge-packet size distribution
(see Fig~\ref{fig: packet size}).

The forms of $\pi^{(\rm{E})}$ and $F^{(\rm{E})}(\cdot)$ are given by
\begin{align}
\label{eq: pi^E}
\pi^{(\rm{E})} &= 
\cfrac{1} 
{
\displaystyle\sum_{s=0}^\infty \, \int_{s \, \ell^{(\rm{d})}}^\infty \, dF^{(\rm{m})}(x)
}
=
\cfrac{1} 
{
\displaystyle\sum_{s=0}^\infty \, \left(1 - F^{(\rm{m})}(s \, \ell^{(\rm{d})})\right)
},
\end{align}
and
\begin{align}
\label{eq: F^E}
F^{(\rm{E})}(x) 
&= \begin{cases}
0, & 0 \le x < \ell^{(\rm{h})}, \\
\displaystyle \sum_{s = 0}^\infty
\left(F^{(\rm{m})}(x + s \, \ell^{(\rm{d})} - \ell^{(\rm{h})}) 
- F^{(\rm{m})}(s \, \ell^{(\rm{d})} - \ell^{(\rm{h})})\right),
& \ell^{(\rm{h})} \le x \le \ell^{(\rm{d})} + \ell^{(\rm{h})}, \\
1, & x > \ell^{(\rm{d})} + \ell^{(\rm{h})}.
\end{cases}
\end{align}

From the definition of body packets, 
the stationary body-packet-size distribution is clearly given by $\textbf{1}(x - \ell^{(\rm{d})} - \ell^{(\rm{h})})$. 
The occurrence probabilities of body packets and edge packets are given by $1 - \pi^{(\textrm{E})}$ and $\pi^{(\textrm{E})}$, respectively.

Letting $\ell^{(\rm{p})}(\stackrel{\triangle}{=} E[L^{(\rm{p})}_\kappa])$
be the mean of stationary distribution of $\{L^{(\rm{p})}_\kappa\}$,
from 
we have
\begin{align}
\ell^{(\rm{p})} &= \pi^{(\textrm{E})} \, \ell^{(\rm{m})} + \ell^{(\rm{h})}.
\label{eq: ell p}
\end{align}

\begin{example}{\it Case of discrete message-size distribution\cite[Example~4]{IKE12_PerEva}.} 
\label{example: Case of discrete message-size distributions}
Consider the case where the message-size distribution function $F^{(m)}(\cdot)$
is given by 
\begin{align}
\label{eq: F m_d}
F^{(\rm{m})}(x) &=\sum_{i=1}^{n_{\rm{d}}} \omega^{(\rm{m})}_i \textbf{1}\left(x - \ell^{(\rm{m})}_i\right),
\end{align}
where $n_d \ge 1$, $w^{(\rm{m})}_i > 0$, $\ell^{(\rm{m})}_i > 0$
for $i = 1, 2, \cdots, n_{\rm{d}}$, and $\sum_{i=1}^{n_{\rm{d}}} w^{(\rm{m})}_i = 1$.

The form of $\pi^{(\rm{E})}$ is given by $\{\sum_{i=1}^{n_{\rm{d}}} w^{(\rm{m})}_i k_i\}^{-1}$ 
with $k_i = \lceil \ell^{(\rm{m})}_i / \ell_d \rceil$.
This can be intuitively shown from the fact
that 1) $k_i$ packets are created from one message of size $\ell^{(\rm{m})}_i$, and 
2) they consist of $k_i - 1$ packets of size $\ell^{(\rm{d})} + \ell^{(\rm{h})}$, that is a body packet,
and one edge packet.
The packet-size distribution can be written as
\begin{align}
\label{eq: F p d}
F^{(\rm{p})}(x) &= 
\left(1 - \pi^{(\rm{E})}\right) \, 
\textbf{1}\left(x - \ell^{(\rm{d})} - \ell^{(\rm{h})}\right)
+ \pi^{(\rm{E})} \, \sum_{i=1}^{n_{\rm{d}}} \,
w^{(\rm{m})}_i \, 
\textbf{1}\left(x - \ell^{(\rm{m})}_i + (k_i - 1) \, \ell^{(\rm{d})} - \ell^{(\rm{h})}\right).
\end{align}
The form of \eqref{eq: F p d} can be rewritten as
\begin{align}
\label{eq: F p d another 1}
F^{(\rm{p})}(x) &\stackrel{\triangle}{=} 
\sum_{i=0}^{n_{\rm{d}}} w^{(\rm{p})}_i \textbf{1}\left(x - \ell^{(\rm{p})}_i\right),
\end{align}
where
\begin{subequations}
\begin{align}
\label{eq: F p d another 2}
&\begin{cases}
w^{(\rm{p})}_0 = 1 - \pi^{(\rm{E})},\\
l^{(\rm{p})}_0 = \ell^{(\rm{d})} +\ell^{(\rm{h})},
\end{cases} \\
&\begin{cases}
w^{(\rm{p})}_i = \pi^{(\textrm{E})} \, w^{(\rm{m})}_i,\\
\label{eq: F p d another 3}
l^{(\rm{p})}_i = \ell^{(\rm{m})}_i - (k_i - 1) \, \ell^{(\rm{d})} + \ell^{(\rm{h})},
\end{cases}i=1, 2, \cdots, n_{\rm{d}}.
\end{align}
\end{subequations}
\hspace*{\fill}~\QED
\end{example}


\section{Goodput analysis}
\label{sec: goodput}

\begin{figure}
\centering{
\epsfig{file=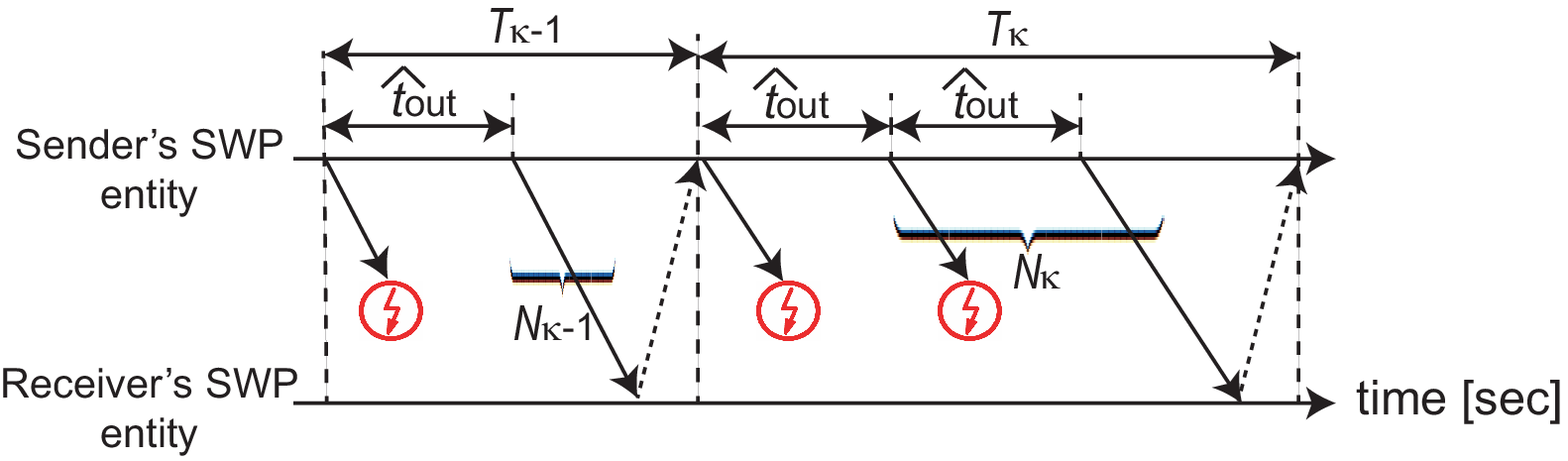, width=10cm} \\
\caption{Example of random variables
$\{T_\kappa\}$ and $\{N_\kappa\}$.}
\label{fig: T N}
}
\end{figure}

In this section,
we derive the form of {\it long-run} goodput.
Let $G_{\rm{p}}$ be long-run goodput of a single SWP connection,
which is defined as the mean number of bits 
by a receiver's SWP entity
per unit time [sec].
To derive the form of $G_{\rm{p}}$, we introduce the following random variables as shown in Fig.~\ref{fig: T N}:
\begin{description}

\item{$T_\kappa$:} time interval between two successful transmissions of the $\kappa$th packet, 
which is expressed in seconds,

\item{$N_\kappa$:} number of transmissions until when the $\kappa$th packet has been successfully transmitted,
$1 \le N_\kappa \le \infty$.

\end{description}

From the theory of renewal reward process \cite[Sec.~3.9]{ROS92},
we have the following proposition.

\begin{proposition}
\label{pro: G}
The form of goodput $G_{\rm{p}}$ is given by
\begin{align}
\label{eq: G}
G_{\rm{p}} &= \cfrac{
\ell^{(\rm{p})} - \ell^{(\rm{h})}
}
{\displaystyle\int_{x=\ell^{(\rm{h})}}^{\ell^{(\rm{d})} + \ell^{(\rm{h})}} \,
E\left[T_\kappa \, | \, L_\kappa^{(\rm{p})} = x\right] \, \, dF^{(\rm{p})}(x)},
\end{align}

\end{proposition}

\begin{proof}
From the Markovian burst bit-error occurrence model mentioned in Section~\ref{sec: burst bit error}
the event point where a sender receives ACK is a renewal point.
Then, the stochastic process $\{T_\kappa\}$ becomes a renewal process.

Suppose that a reward $R_\kappa$ is earned at the time of the $\kappa$~th
renewal,
that is the mean number of bits by a receiver's SWP entity at interval $T_\kappa$.
Letting $Y(s)$ be the total reward earned by time $s$ [sec],
from \cite[Thorem~3.16]{ROS92},
we have
\begin{align}
\label{eq: G 1}
G &= \lim_{s \to \infty} \, \cfrac{Y(s)}{s} = \cfrac{E\left[R_\kappa\right]}{E\left[T_\kappa\right]}.
\end{align}
Here, 
\begin{align}
\label{eq: R}
E\left[R_\kappa\right] &= \displaystyle\int_{x=\ell^{(\rm{h})}}^{\ell^{(\rm{d})} + \ell^{(\rm{h})}} \, 
\Pr\left(L_\kappa^{(\rm{p})} = x\right) \, 
\left(x - \ell^{(\rm{h})}\right) \, dF^{(\rm{p})}(x) \notag \\
&= \ell^{(\rm{p})} - \ell^{(\rm{h})}.
\end{align}
Substitute \eqref{eq: R} into \eqref{eq: G 1},
we obtain \eqref{eq: G}.
\end{proof}

\medskip

The assumption {\tt A2} yields the form of $E[T_\kappa \, | \, L_\kappa^{(\rm{p})}=x, N_\kappa = n]$ 
given by
\begin{align}
\label{eq: E T N}
E\left[T_\kappa \, | \, L_\kappa^{(\rm{p})} =x, N_\kappa = n\right] &=
(n - 1) \, t_{\rm{out}} 
+ \cfrac{x + \ell^{(\rm{ACK})}}{\mu_{\rm{c}}} + t_{\rm{pro}}, 
\quad n = 1, 2, \cdots.
\end{align}

\begin{remark}
Equation~\eqref{eq: E T N} shows that overhead of
$\cfrac{\ell^{(\rm{h})} + \ell^{(\rm{ACK})}}{\mu_{\rm{c}}} + t_{\rm{pro}}$ per one packet transmission
is at least necessary
for one transmission of a packet with size equaled to $x$.
\end{remark}

To derive the form of $\Pr(N_\kappa = n \,|\, L_\kappa^{(\rm{p})}= x)$,
we introduce the following events:
\begin{description}

\item{\ $\mathcal{E}_{\rm{s}}(t, x, \eta)$:} event that satisfying:

\begin{itemize}

\item the transmission of a packet of $x$ bits was started at time $t$ [bit], 

\item it has been finished successfully, i.e., with no bit errors, and

\item the link state is $\eta$ at the time just after the transmission,
i.e., $\upsilon(t+x)=\eta$.

\end{itemize}

\item{\ $\mathcal{E}_{\rm{f}}(t, x, \eta)$:} 
the same event as $\mathcal{E}_{\rm{s}}(t, x, \eta)$,
except that the transmission failed
(i.e., the packet contains one or more erroneous bits). 
\end{description}

Then, we have the following lemma.

\begin{lemma}
\label{lemma: Ps Pf}

Let $\mathbf{P}_{\rm{s}}(x)$ and $\mathbf{P}_{\rm{f}}(x)$ be matrices 
whose $(\xi, \eta)$th entries are
conditional probabilities of events
$\mathcal{E}_{\rm{s}}(t, x, \eta)$ and $\mathcal{E}_{\rm{f}}(t, x, \eta)$, 
respectively, given that $\upsilon(t)=\xi$.
Then, we have
\begin{align}
\label{eq: P_s}
\mathbf{P}_{\rm{s}}(x) &\stackrel{\triangle}{=}
\left[
\Pr\Big(\mathcal{E}_{\rm{s}}(t, x, \eta) \,|\, \upsilon(t)=\xi\Big),
\xi \in \mathcal{C}, \eta \in \mathcal{C}
\right]\notag\\
&= \mathbf{Q}^x, \\
\intertext{and} 
\label{eq: P_f}
\mathbf{P}_{\rm{f}}(x) &\stackrel{\triangle}{=}
\left[
\Pr\left(\mathcal{E}_{\rm{f}}(t, x, \eta) \,|\, \upsilon(t)=\xi\right),
\xi \in \mathcal{C}, \eta \in \mathcal{C}
\right]\notag\\
&= 
\left[
\Pr\left(\upsilon(t + x)=\eta \,|\, \upsilon(t)=\xi\right)
- \Pr\left(\mathcal{E}_{\rm{s}}(t, x, \eta) \,|\, \upsilon(t)=\xi\right),
\xi \in \mathcal{C}, \eta \in \mathcal{C}
\right] \notag\\
&= \mathbf{P}_{\rm{c}}^x - \mathbf{P}_s(x)
= \mathbf{P}_{\rm{c}}^x - \mathbf{Q}^x.
\end{align}
Here,
$\mathbf{Q}$ is defined as
the matrix 
of which $(\xi, \eta)$th entry is the conditional probability that 
the link state changes to $\eta \in \mathcal{C}$ 
after a successful bit transmission,
given that it was $\xi \in \mathcal{C}$:
\begin{align}
\label{eq: Q}
\mathbf{Q} &\stackrel{\triangle}{=}
\left[\Pr\left(\upsilon(t+1)=\eta \quad\text{and no bit error occurs}\, | \, \upsilon(t)=\xi\right), \eta \in \mathcal{C}, \xi \in \mathcal{C}\right] \notag \\
%
%
&=\begin{pmatrix}
\left(1 - p^{(\rm{G})}\right) \, \left(1 - \lambda\right) & 
\left(1 - p^{(\rm{G})}\right) \, \lambda \\
\left(1 - p^{(\rm{B})}\right) \, \gamma & 
\left(1 - p^{(\rm{B})}\right) \, \left(1 - \gamma\right)
\end{pmatrix}.
\end{align}
\end{lemma}

\medskip

\begin{proof}
From the definitions of $\mathbf{P}_{\rm{s}}(x)$ and $\mathbf{P}_{\rm{f}}(x)$, 
\eqref{eq: P_s} and \eqref{eq: P_f} are clear.
\end{proof}

\bigskip

\begin{proposition}
\label{pro: pr N L}
The form of $\Pr(N_\kappa = n \,|\, L_\kappa^{(\rm{p})} = x)$ is given by
\begin{align}
\label{eq: pr N L}
\Pr\left(N_\kappa = n \,|\, L_\kappa^{(\rm{p})} = x\right) &= 
\mathbf{\pi}_{\rm{c}} \, \left\{\mathbf{S}(x, t_{\rm{out}})\right\}^{n - 1} \, 
\mathbf{Q}^{x} \, \mathbf{e}, \quad n = 1, 2, \cdots, 
\end{align}
where $\mathbf{e} \stackrel{\triangle}{=} (1, 1)^T$ is a unit vector and
\begin{align}
\label{eq: S}
\mathbf{S}(x, t_{\rm{out}})
&\stackrel{\triangle}{=} 
\mathbf{P}_{\rm{c}}^{t_{\rm{out}}} - \mathbf{Q}^{x} \mathbf{P}_{\rm{c}}^{t_{\rm{out}} - x}. 
\end{align}

\end{proposition}

\medskip

\begin{proof}
We consider the following scenario:
\begin{itemize}

\item the $\kappa$~th packet whose size is $x$, i.e., $L_\kappa^{(\rm{p})} = x$,
was transmitted first time at time $t_0$ [bit], and

\item it has been transmitted unsuccessfully $n (\ge0)$ times and transmitted successfully
in the $(n+1)$~th transmission,
i.e., in the $n$~th retransmission, resulting in $N_\kappa = n$.
\end{itemize}

An example of a packet transmission sequence,
$\{\xi_i\}$ where $\xi = \upsilon(t_0 + i \, t_{\rm out})$ and
$\{\eta_i\}$ where $\eta = \upsilon(t_0 + i \, t_{\rm out} + x)$
when the packet of the second retransmission
has been transmitted successfully is shown in Fig.~\ref{fig: packet}.

From Lemma~\ref{lemma: Ps Pf} and Fig.~\ref{fig: packet},
the form of $\Pr(N_\kappa = n \,|\, L_\kappa^{(\rm{p})} = x)$ in this case
is given by
\begin{align}
&\Pr\left(N_\kappa = n \,|\, L_\kappa^{(\rm{p})} = x\right)
\notag \\
&= 
\sum_{\xi_0 \in \mathcal{C}} 
\sum_{\eta_0 \in \mathcal{C}} 
\sum_{\xi_1 \in \mathcal{C}} 
\sum_{\eta_1 \in \mathcal{C}} 
\sum_{\xi_2 \in \mathcal{C}}
\cdots
\sum_{\eta_{n-1} \in \mathcal{C}}
\sum_{\xi_n \in \mathcal{C}} \sum_{\eta_n \in \mathcal{C}} \notag \\
&\quad \Pr \Big(
\{\upsilon(t_0) = \xi_0\} \notag \\
&\quad
\cap \mathcal{E}_{\rm{f}}(t_0, x, \eta_0)\notag \\
&\quad
\cap \{\upsilon(t_0+x) = \eta_0\} \cap \{\upsilon(t_0+t_{\rm{out}}) = \xi_1\} \notag \\
&\quad 
\cap \mathcal{E}_{\rm{f}}(t_0+t_{\rm{out}}, x, \eta_1) \notag \\
&\quad
\cap \{\upsilon(t_0+t_{\rm{out}} + x) = \eta_1\}\cap \{\upsilon(t_0+2 \, t_{\rm{out}}) = \xi_2\}\notag \\
&\quad 
\cap \cdots\cdots\cdots\cdots\cdots\cdots \notag \\
&\quad 
\cap \mathcal{E}_{\rm{f}}(t_0+(n-1) \, t_{\rm{out}}, x, \eta_{n-1})\notag \\
&\quad
\cap \{\upsilon(t_0+(n-1) \, t_{\rm{out}} + x) = \eta_{n-1}\} 
\cap \{\upsilon(t_0+ n \, t_{\rm{out}}) = \xi_n\}
\notag \\
&\quad 
\cap \mathcal{E}_{\rm{s}} (t_0+n \, t_{\rm{out}}, x, \eta_n)
\Big) \notag \\
&=\mathbf{\pi}_{\rm{c}} 
\left\{
\mathbf{P}_{\rm{f}}(x) \mathbf{P}_{\rm{c}}^{t_{\rm{out}} - x}
\right\}^n
\mathbf{P}_s(x) \, \mathbf{e} \notag \\
&=\mathbf{\pi}_{\rm{c}} 
\left\{
(\mathbf{P}_{\rm{c}}^{x} - \mathbf{Q}^{x})
\mathbf{P}_{\rm{c}}^{t_{\rm{out}}- x}
\right\}^n
\mathbf{Q}^{x} \, \mathbf{e}
\notag \\
&= \mathbf{\pi}_{\rm{c}} \, \left\{\mathbf{S}(x, t_{\rm{out}})\right\}^n \, \mathbf{Q}^x \, \mathbf{e},
\end{align}
which is \eqref{eq: pr N L}.
\end{proof}

\begin{figure}
\centering{
\epsfig{file=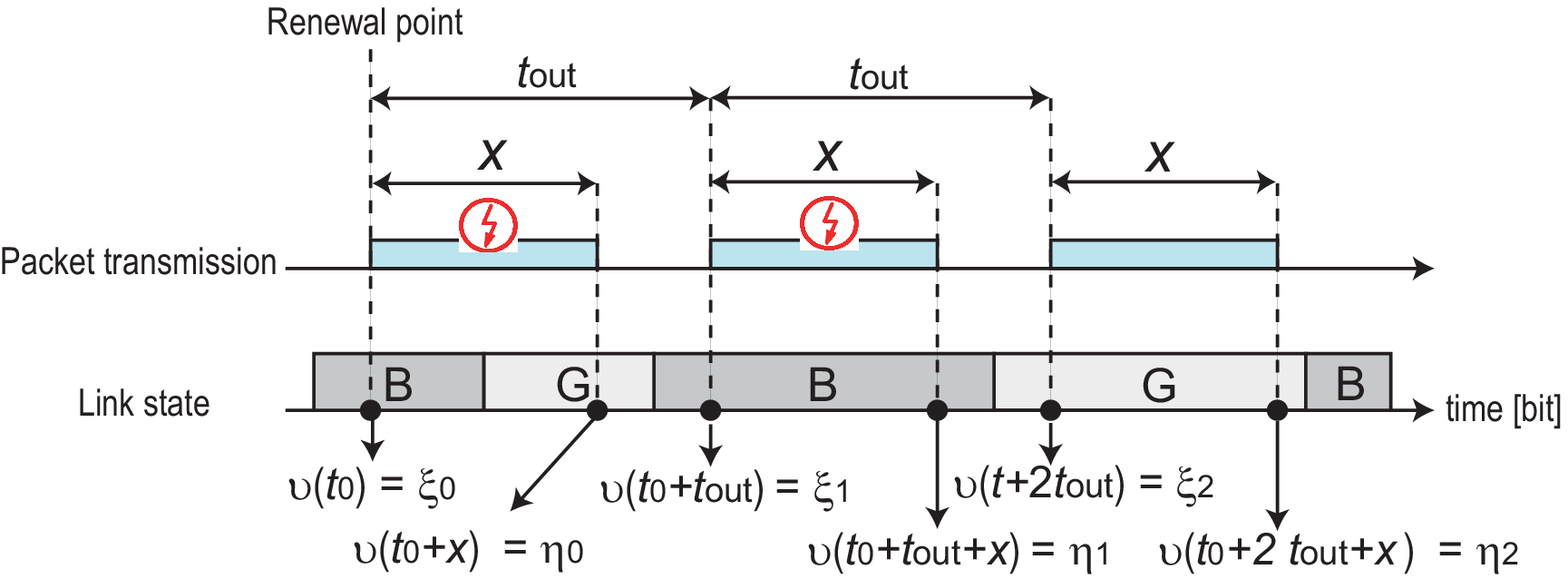, width=13cm} \\
\caption{Example of link state sequence
$\{\xi_i\}$ and $\{\eta_i\}$.
In this example,
$\{\xi_i: i = 0, 1, 2\} = \{\textrm{B}, \textrm{B}, \textrm{G}\}$
and
$\{\eta_i: i = 0, 1, 2\} = \{\textrm{G}, \textrm{B}, \textrm{G}\}$.}
\label{fig: packet}}
\end{figure}

\bigskip

\begin{proposition}
\label{pro: pr N L i i d}
%
Let us consider the \textit{i.i.d.}~bit-error model with $\lambda+\gamma=1$,
which is a special case described in Remark~\ref{remark: i.i.d. bit error model}.
In this case, we have
\begin{align}
\label{eq: pr N L iid}
\Pr(N_\kappa = n \,|\, L_\kappa^{(\rm{p})} = x) &=
\left(1 - h(x)\right) \, \left\{h(x)\right\}^{n - 1}, \quad & n = 1, 2, \cdots, 
\end{align}
where 
\begin{align}
\label{eq: h}
h(x) &= 1 - \left(1 - p_{\rm{e}}\right)^x.
\end{align}

\end{proposition}

\medskip

\begin{proof}
From $\lambda+\gamma=1$, we have
\begin{align}
\label{eq: P^k}
\mathbf{P}_{\rm{c}} &=
\begin{pmatrix}
\pi^{(\textrm{G})} & \pi^{(\textrm{B})} \\
\pi^{(\textrm{G})} & \pi^{(\textrm{B})}
\end{pmatrix}.
\end{align}
From $\pi^{(\textrm{G})} + \pi^{(\textrm{B})} = 1$,
$\mathbf{P}_{\rm{c}}^k$ for $k = 1, 2, \cdots,$ is given by
\begin{align}
\mathbf{P}_{\rm{c}}^k &= \mathbf{P}_{\rm{c}}, \quad\quad k= 1, 2, \cdots.
\end{align}
Furthermore, we have
\begin{align}
\label{eq: Q^k}
\mathbf{Q}^k &= (1 - \hat{p}_e)^{k-1} \mathbf{Q}.
\end{align}

Because of the definition of $\mathbf{\pi}_{\rm{c}}$,
that is the stationary probability vector of $\mathbf{P}_{\rm{c}}$,
we obtain
\begin{align}
\label{eq: pi_P_k}
\mathbf{\pi}_{\rm{c}} &= \mathbf{\pi}_{\rm{c}} \, \mathbf{P}_\textrm{c}^k, \quad\quad k= 1, 2, \cdots.
\end{align}
Because $1 - \hat{p}_e$ is the eigen-value
corresponding to the eigen-vector $\mathbf{\pi}_{\rm{c}}$ of matrix $\mathbf{Q}$,
we have
\begin{align}
\label{eq: pi_Q}
\mathbf{\pi}_{\rm{c}} \, \mathbf{Q} &= (1 - \hat{p}_e) \, \mathbf{\pi}_{\rm{c}}.
\end{align}
Hence, we obtain
\begin{align}
\label{eq: pi_Q_k}
\mathbf{\pi}_{\rm{c}} \, \mathbf{Q}^k &= (1 - \hat{p}_e)^k \, \mathbf{\pi}_{\rm{c}}, \quad\quad k= 1, 2, \cdots.
\end{align}

The form of $\mathbf{\pi}_{\rm{c}} \, \mathbf{S}(x, t_{\rm{out}})$ is given by
%
\begin{align}
\label{eq: pi_c_S}
\mathbf{\pi}_{\rm{c}} \, \mathbf{S}(x, t_{\rm{out}})
&=
\mathbf{\pi}_{\rm{c}} \, 
\left(\mathbf{P}_{\rm{c}}^{t_{\rm{out}}} - \mathbf{Q}^{x} \, \mathbf{P}_{\rm{c}}^{t_{\rm{out}} - x}\right)
&&\text{(from \eqref{eq: S})} \notag \\ 
%
&=
\mathbf{\pi}_{\rm{c}} \, \mathbf{P}_{\rm{c}}^{t_{\rm{out}}}
- \mathbf{\pi}_{\rm{c}} \, \mathbf{Q}^{x} \, \mathbf{P}_{\rm{c}}^{t_{\rm{out}} - x} \notag \\ 
%
&= \mathbf{\pi}_{\rm{c}} \, \mathbf{P}_\textrm{c} 
- \mathbf{\pi}_{\rm{c}} \, \mathbf{Q}^{x} \, \mathbf{P}_{\rm{c}}^{t_{\rm{out}} - x} 
&&\text{(from \eqref{eq: pi_P_k})} \notag \\ 
%
&= \mathbf{\pi}_{\rm{c}} \, \mathbf{P}_\textrm{c} - (1 - \hat{p}_e)^{x-1} \, 
\mathbf{\pi}_{\rm{c}} \, \mathbf{Q} \, \mathbf{P}_{\rm{c}}
&&\text{(from \eqref{eq: pi_Q_k})} \notag \\
%
&= \mathbf{\pi}_{\rm{c}} \, \mathbf{P}_\textrm{c} - (1 - \hat{p}_e)^{x} \, 
\mathbf{\pi}_{\rm{c}} \, \mathbf{P}_{\rm{c}} 
&&\text{(from \eqref{eq: pi_Q})} \notag \\
%
%
&= \{1 - (1 - \hat{p}_e)^{x}\} \, \mathbf{\pi}_{\rm{c}} \, \mathbf{P}_{\rm{c}} \notag \\
%
&= h(x) \, \mathbf{\pi}_{\rm{c}} \, \mathbf{P}_{\rm{c}} 
&&\text{(from \eqref{eq: h})} \notag \\
&= h(x) \, \mathbf{\pi}_{\rm{c}}.
&&\text{(from \eqref{eq: pi_P_k})}
\end{align}
Then, in a recursive manner,
we obtain
\begin{align}
\label{eq: pi S n}
\mathbf{\pi}_{\rm{c}} \, \, \mathbf{S}(x, t_{\rm{out}})^n &=
\left\{h(x)\right\}^n \, \mathbf{\pi}_{\rm{c}}, \quad\quad n = 2, 3, \cdots,
\end{align}

On the other hand,
we have
\begin{align}
\label{eq: pi Q e}
\mathbf{\pi}_{\rm{c}} \, \mathbf{Q}^{x} \, \mathbf{e}
%
&= (1 - \hat{p}_e)^{x-1} \, \mathbf{\pi}_{\rm{c}} \, \mathbf{Q} \, \mathbf{e} \notag \\
%
&= (1 - \hat{p}_e)^{x-1} \, (1 - \hat{p}_e) \, \mathbf{\pi}_{\rm{c}} \, \mathbf{e} 
&&\text{(from \eqref{eq: pi_Q})} \notag \\ 
&= (1 - \hat{p}_e)^{x}.
\end{align}

Substituting \eqref{eq: pi_c_S} and \eqref{eq: pi Q e} 
into \eqref{eq: pr N L},
we obtain \eqref{eq: pr N L iid}.

\end{proof}


\begin{proposition}
\label{pro: G infinite}
The form of $E[T_\kappa \, | \, L_\kappa^{(\rm{p})} = x]$ in \eqref{eq: G} is given by
\begin{align}
\label{eq: E T L}
E\left[T_\kappa\, | \, L_\kappa^{(\rm{p})} = x\right] &=
t_{\rm{out}} \, \left[\mathbf{\pi}_{\rm{c}} 
\left(\mathbf{I} - \mathbf{S}(x, t_{\rm{out}})\right)^{-1} \, \mathbf{e} - 1\right]
+ \cfrac{x + \ell^{(\rm{ACK})}}{\mu_{\rm{c}}} + t_{\rm{pro}}. 
\end{align}
\end{proposition}

\medskip

\begin{proof}
The form of $E\left[T_\kappa\, | \, L_\kappa^{(\rm{p})} = x\right]$
is given by
\begin{align}
\label{eq: tilde E T}
E\left[T_\kappa\, | \, L_\kappa^{(\rm{p})} = x\right] &= 
\sum_{n=1}^\infty \, \Pr(N_\kappa = n \,|\, L_\kappa^{(\rm{p})} = x) \, 
E\left[T_\kappa \, | \, L_\kappa^{(\rm{p})} =x, N_\kappa = n \right] \notag \\
&=\sum_{n=1}^\infty \, \Pr(N_\kappa = n \,|\, L_\kappa^{(\rm{p})} = x) \, 
\left[(n -1) \, \hat{t}_{\rm{out}}
+ \cfrac{x + \ell^{(\rm{ACK}}}{\mu_{\rm{c}}}\right] 
+ t_{\rm{pro}} \notag \\
&= \hat{t}_{\rm{out}} \, \sum_{n=1}^\infty \, (n -1) \, \Pr(N_\kappa = n \,|\, L_\kappa^{(\rm{p})} = x)
+ \cfrac{x + \ell^{(\rm{ACK}}}{\mu_{\rm{c}}} + t_{\rm{pro}}.
\end{align}
From \eqref{eq: pr N L},
the form $\sum_{n=1}^\infty \, (n -1) \, \Pr(N_\kappa = n \,|\, L_\kappa^{(\rm{p})} = x)
(= E[N_\kappa \, | \, L_\kappa^{(\rm{p})} = x] - 1)$ in \eqref{eq: tilde E T})
can be written as
\begin{align}
\label{eq: Pr N L ap}
\sum_{n=1}^\infty \, (n -1) \, \Pr(N_\kappa = n \,|\, L_\kappa^{(\rm{p})} = x)
&= \mathbf{\pi}_{\rm{c}} \,
\left[
\sum_{n=1}^{\infty} (n - 1)\, \mathbf{S}(x, t_{\rm{out}})^{n-1}
\right]\,
\mathbf{Q}^x \, \mathbf{e}
\notag \\
&= \mathbf{\pi}_{\rm{c}} \, \mathbf{S}(x, t_{\rm{out}}) 
\big(
\mathbf{I} - \mathbf{S}(x, t_{\rm{out}})
\big)^{-2} \mathbf{Q}^x \, \mathbf{e}.
\end{align}
Because $\mathbf{P}_{\rm{c}}$ is a probability matrix,
we obtain
\begin{align}
\label{eq: P_c e}
\mathbf{P}_{\rm{c}}^{t_{\rm{out}}} \, \mathbf{e} &= 
\mathbf{P}_{\rm{c}}^{t_{\rm{out}} - x} \, \mathbf{e} = \mathbf{e}.
\end{align}
From \eqref{eq: P_c e},
$\mathbf{S}(x, t_{\rm{out}}) \, \mathbf{e}$ is given by 
\begin{align}
\label{eq: R e}
\mathbf{S}(x, t_{\rm{out}}) \, \mathbf{e} &=
\left(
\mathbf{P}_{\rm{c}}^{t_{\rm{out}}} - \mathbf{Q}^x \, \mathbf{P}_{\rm{c}}^{t_{\rm{out}}-x}\right) \, 
\mathbf{e} \notag \\
&= \mathbf{e} - \mathbf{Q}^x \, \mathbf{e}.
\end{align}
Equation~\eqref{eq: R e} can be changed into
\begin{align}
\label{eq: Q e}
\mathbf{Q}^x \, \mathbf{e} 
&= \mathbf{e} - \mathbf{S}(x, t_{\rm{out}}) \, \mathbf{e} \notag \\
&= \left(\mathbf{I} - \mathbf{S}(x, t_{\rm{out}})\right) \, \mathbf{e}.
\end{align}
Substitution of \eqref{eq: Q e} into \eqref{eq: Pr N L ap} yields
\begin{align}
\label{eq: ap3 last}
\sum_{n=1}^\infty \, (n -1) \, \Pr(N_\kappa = n \,|\, L_\kappa^{(\rm{p})} = x) 
&= \mathbf{\pi}_{\rm{c}} \, \mathbf{S}(x, t_{\rm{out}})
\left(
\mathbf{I} - \mathbf{S}(x, t_{\rm{out}})
\right)^{-1} \, \mathbf{e} \notag \\
&= \mathbf{\pi}_{\rm{c}} \, \left[\mathbf{I} - (\mathbf{I} - \mathbf{S}(x, t_{\rm{out}}) \right]
\left(
\mathbf{I} - \mathbf{S}(x, t_{\rm{out}})
\right)^{-1} \, \mathbf{e} \notag \\
&= \mathbf{\pi}_{\rm{c}} \left(
\mathbf{I} - \mathbf{S}(x, t_{\rm{out}})
\right)^{-1} \, \mathbf{e} - \mathbf{\pi}_{\rm{c}} \, \mathbf{I} \, \mathbf{e}\notag \\
&= \mathbf{\pi}_{\rm{c}} \left(
\mathbf{I} - \mathbf{S}(x, t_{\rm{out}})
\right)^{-1} \, \mathbf{e} - 1.
\end{align}
By substituting \eqref{eq: ap3 last} into \eqref{eq: tilde E T},
we obtain \eqref{eq: E T L}.
\end{proof}

\begin{remark}
The term of $\mathbf{\pi}_{\rm{c}} (\mathbf{I} - \mathbf{S}(x, t_{\rm{out}}))^{-1} \, \mathbf{e} - 1$
in \eqref{eq: E T L}
implies the mean retransmission times of the packet whose size is $x$,
i.e., $E[N_\kappa\, | \, L_\kappa^{(\rm{p})} = x] - 1$.
\end{remark}


\section{Numerical results and discussions}
\label{sec: results}

In this section,
we answer the research questions mentioned in section~\ref{sec: Introduction}.
First,
we explain the assumptions made for numerical results.
Next, we validate the constant packet-size assumption
to answer research question Q1.
In final, we discuss the relationship between payload size and goodput under several bit-error link conditions
to answer research question Q2.

\subsection{Assumptions}

In the following numerical results,
we assume that 
1) link capacity $\mu_{\rm{c}}$: $1$~Mbps, and 
2) Bit-error rates of link state $\{\rm{G}, \rm{B}\}$: $p^{(\rm{G})} = 0$, $p^{(\rm{B})}= 1$,
3) header size: $\ell^{(\rm{h})}$ (= $\ell^{(\rm{ACK}}$) = 38~bytes,
4) timeout value $\hat{t}_{\rm{out}}$: $100$~msec, and
5) propagation delay $t_{\rm{pro}}$: $1$~msec.

To answer research questions comprehensively, 
we consider the simple case, i.e., case of constant message size.
Thus,
message sizes are assumed to be distributed according to $F^{(\rm{m})}(x)$ given by
\begin{align}
\label{eq: F m constant}
F^{(\rm{m})}(x) &= \textbf{1}\left(x - \ell^{(m)}\right).
\end{align}
The constant message size $\ell^{(m)}$ is assumed to be
$4000$~bytes, 
which is the approximate measured mean message size of Web object \cite{MOL00}.

\subsection{Validation of constant packet-size assumption}
\label{sec: validation}

From the case of $n_{\rm{d}} = 1$, $w^{(\rm{p})}_1$, and $l^{(\rm{p})}_1 = \ell^{(m)}$ for example~\ref{example: Case of discrete message-size distributions},
the packet-size distribution $F^{(\rm{p})}(x)$ is given by
\begin{align}
F^{(\rm{p})}(x) &=
\left(1 - \pi^{(\rm{E})}\right) \, \textbf{1} \left(x - \ell^{(\rm{d})} - \ell^{(\rm{h})} \right)
+ \pi^{(\rm{E})} \, \textbf{1} \left(x - \left\{\ell^{(m)} - \left(r_{\rm{c}} - 1\right) \, \ell^{(\rm{d})} + \ell^{(\rm{h})}\right\}\right), \notag \\
\pi^{(\rm{E})}&= \left\lceil \cfrac{\ell^{(\rm{m})}}{\ell^{(\rm{d})}} \right\rceil.
%
\end{align}
%

\begin{figure}
\centering
\epsfig{file=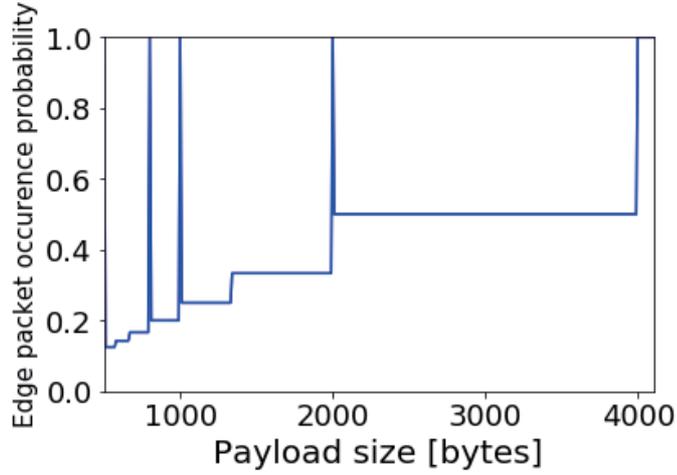, width=9cm} \\
\caption{Edge-packet occurrence probability $\pi^{(\rm{E})}$ versus payload size $\ell^{(\rm{d})}$.}
\label{fig: edge}
\end{figure}

Figure~\ref{fig: edge} shows edge-packet occurrence probability $\pi^{(\rm{E})}$ versus payload size $\ell^{(\rm{d})}$.
From this figure, edge-packet occurrence probability $\pi^{(\rm{E})}$ has the stepwise tendency.
We note that the value of $\pi^{(\rm{E})}$ is one but
all packets are constant in size
when message size is a multiple of payload size, such as 500, 1000, 2000, and 4000~bytes.

\begin{figure}
\centering
\epsfig{file=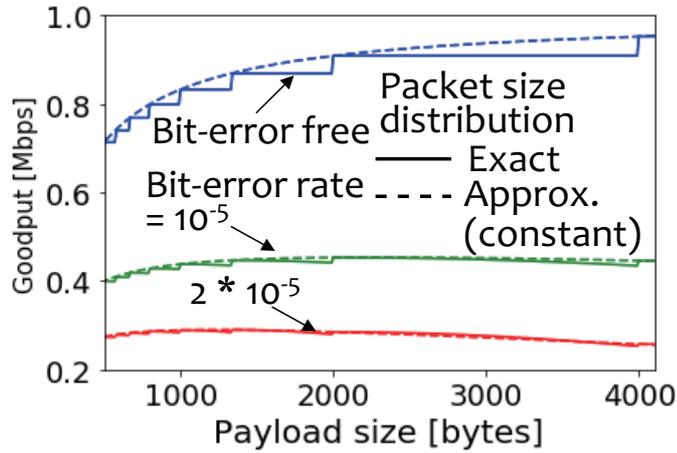, width=9cm} \\
\caption{Goodput $G_{\rm{p}}$ versus payload size $\ell^{(\rm{d})}$
for different bit-error rates.}
\label{fig: exact}
\end{figure}

Figure~\ref{fig: exact} shows goodput $G_{\rm{p}}$ versus payload size $\ell^{(\rm{d})}$
for different bit-error rates $p_{\rm{e}}$.
In this figure, dashed line means goodput using constant packet-size approximation.
From this figure this approximation error is not negligible under low bit-error rate.
On the other hand,
under high bit-error rate,
this error is negligible.
The reason for this is that
the number of retransmissions of packet whose size is $\ell^{(\rm{d})}+\ell^{(\rm{h})}$, 
which is the maximum packet size,
are dominant in the total number of transmissions of packets
because the packet is more likely to be corrupted (or retransmitted)
if the packet size is larger 
\cite{IKE05_RPSP_MSWiM}.

\subsection{Impact of payload size to goodput
in case of independent bit errors}

\begin{figure}
\centering
\epsfig{file=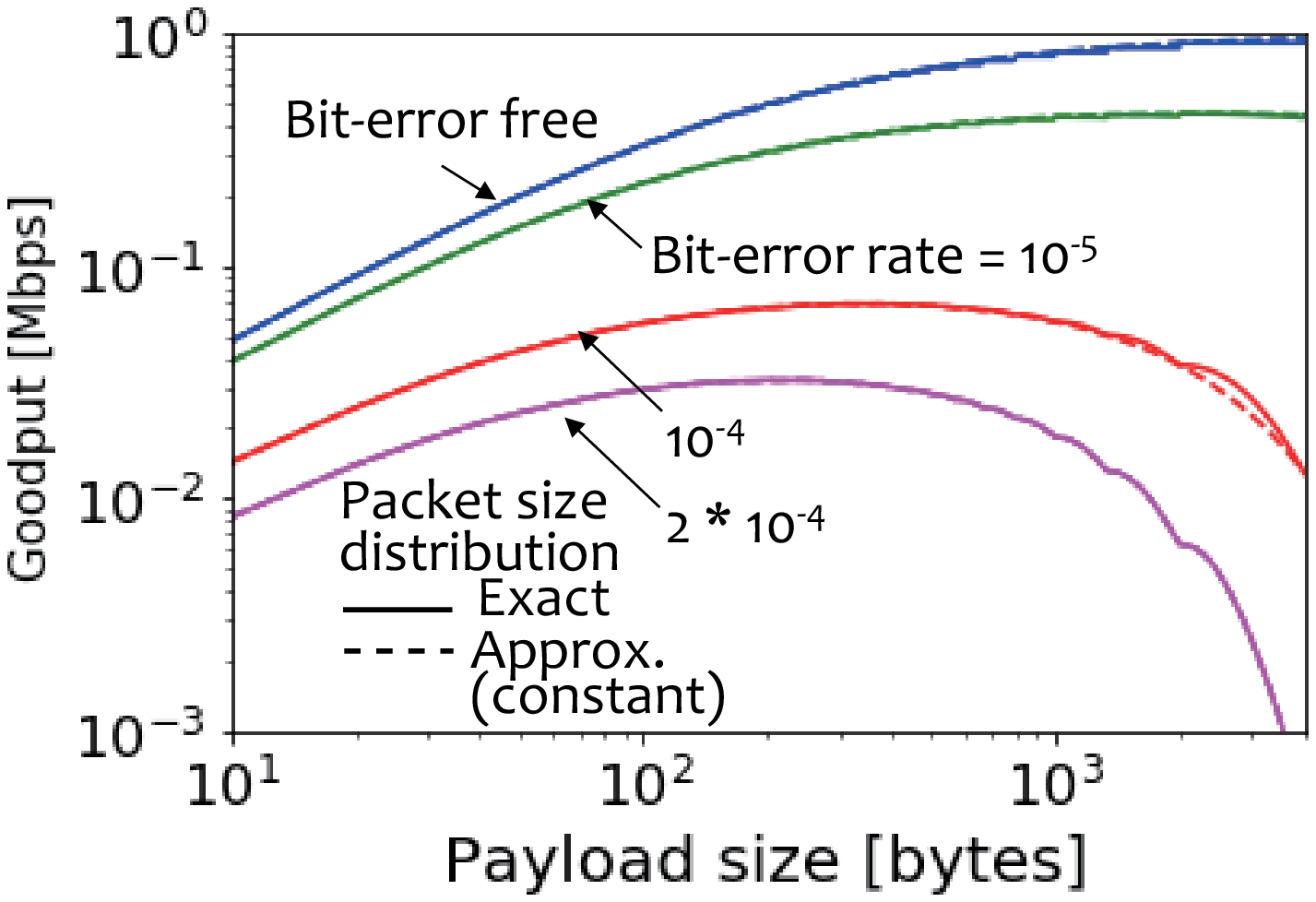, width=9cm} \\
\caption{Goodput $G_{\rm{p}}$ versus payload size $\ell^{(\rm{d})}$
for different bit-error rates $p_{\rm{e}}$
in case of i.i.d. bit errors.}
\label{fig: iid}
\end{figure}

Figure~\ref{fig: iid}
shows goodput $G_{\rm{p}}$ versus payload size $\ell^{(\rm{d})}$
for different bit-error rates $p_{\rm{e}}$
in case of i.i.d. bit errors.
This figure shows
that the curves of goodput are concave
under high bit-error rates.
The reason for this is as follows:

\begin{itemize}

\item If payload size is enough small, 
the communication network is operating in efficiency
because an overheating transmitting PCI and acknowledgements is not negligible.

\item On the other hand,
as payload size increases, 
the packet is more likely to be corrupted.
Hence,
it causes more retransmissions, resulting in reduced goodput. 

\end{itemize}

From the above argument,
there exists an optimum payload size
in the sense of maximizing goodput.

\subsection{Impact of payload size to goodput
in case of burst bit errors}

\begin{figure}
\centering
\epsfig{file=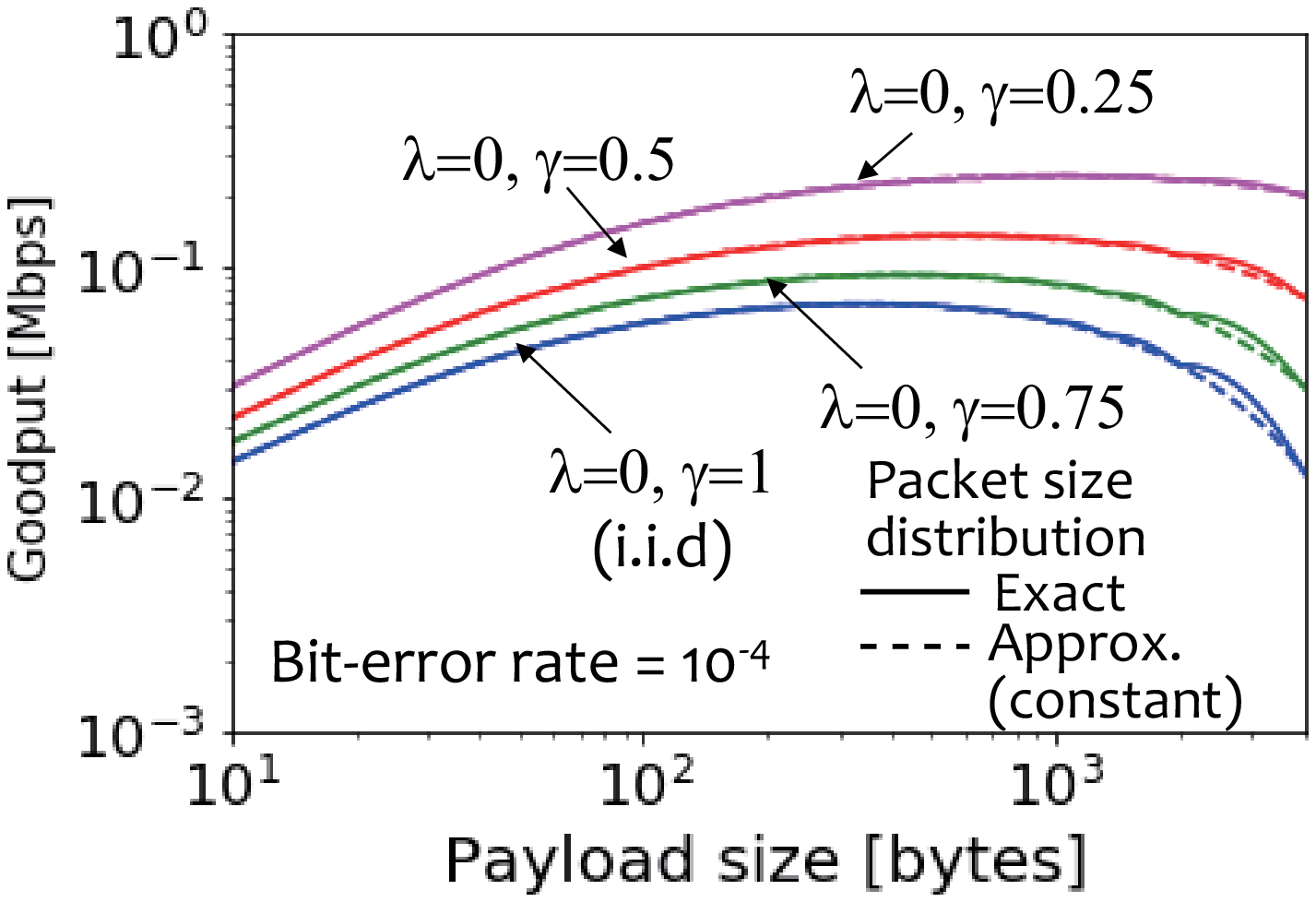, width=9cm} \\
\caption{Goodput $G_{\rm{p}}$ versus payload size $\ell^{(\rm{d})}$
for different mean bit-error burst lengths $\gamma^{-1}$
in the case of mean bit-error rate $\hat{p}_{\rm{e}}$ equaled to $10^{-4}$.}
\label{fig: burst}
\end{figure}

Final, we investigate the effect of mean bit-error burst length on goodput.
Figure~\ref{fig: burst} 
shows goodput $G_{\rm{p}}$ versus payload size $\ell^{(\rm{d})}$
for different mean bit-error burst lengths $\gamma^{-1}$
in the case of mean bit-error rate $\hat{p}_{\rm{e}}$ equaled to $10^{-4}$.
This figure shows that
the larger mean bit-error burst length yields less concave curves of goodput.


\section{Conclusion}
\label{sec: conclusion}

This paper investigated the effect of payload size
on goodput for wireless networks where message segmentations occur
and corrupted packets are recovered by a stop-and-wait protocol.
To achieve this,
we derived the analytical form of goodput for wireless networks
using segmented packet-size distribution,
given a message-size distribution and a payload size.

From numerical results,
we show that the traditional constant packet-size assumption is not justified in the case of low bit error rates.
Furthermore,
we showed that the curves of goodput are concave in payload size
under high bit error rates.
In addition,
we indicated that the larger mean bit-error burst length yields less concave curves of goodput.

The remaining issues include 
the investigation of the effect of payload size
on goodput for various message size distributions,
the extension of our model to realistic protocols rather than a simple stop-and wait protocol,
and the development of simple and cost-effective payload adaptation algorithm
in a wireless environment whose condition is dynamically changed.

\section*{Acknowledgment}

This work was supported in part by JSPS KAKENHI Grant Number JP15K00139.


\bibliographystyle{IEEEtran}
\bibliography{paper}

\end{document}